\documentclass[11pt]{article}
\usepackage[colorlinks, citecolor=blue, linkcolor=black]{hyperref}
\usepackage{fullpage,amssymb,amsmath,amsthm}
\usepackage[linesnumbered,boxed,noend]{algorithm2e}

\newtheorem{THEOREM}{Theorem}

\newtheorem{theorem}{Theorem}[section]

\newtheorem{claim}[theorem]{Claim}

\newtheorem{corollary}[theorem]{Corollary}

\newtheorem{definition}[theorem]{Definition}

\newtheorem{lemma}[theorem]{Lemma}

\newtheorem{observation}[theorem]{Observation}
\newtheorem{fact}[theorem]{Fact}
\newtheorem{problem}[theorem]{Problem}

\newtheorem{proposition}[theorem]{Proposition}

\newcommand{\set}[2][]{{\left\{#2\right\}}^{#1}}
\newcommand{\bools}[1]{\set[#1]{0,1}}

\newcommand{\abs}[1]{\left|{#1}\right|}
\newcommand{\size}[1]{\abs{#1}}

\newcommand{\condset}[2]{\set{#1 \; \left| \;#2 \right. }}

\newcommand{\restrict}[1]{| _{#1}}

\newcommand{\N}{\mathbb{N}}

\newcommand{\F}{\mathbb{F}}
\newcommand{\Ext}{\mathbb{E}}
\newcommand{\Fn}{\mathbb{F}[x_{1},x_{2},\ldots ,x_{n} ]}
\newcommand{\Fnxy}{\mathbb{F}[x_{1},x_{2},\ldots ,x_{n},y_{1},y_{2},\ldots ,y_{n} ]}
\newcommand{\BigO}{\mathcal{O}}
\newcommand{\eqdef}{\stackrel{\Delta}{=}}
\newcommand{\nequiv}{\not \equiv}

\newcommand{\var}{\mathrm{var}}

\newcommand{\poly}{\mathrm{poly}}

\newcommand{\Prop}{\mathcal{P}}

\newcommand{\xb}{\bar{x}}
\newcommand{\yb}{\bar{y}}
\newcommand{\ab}{\bar{a}}
\newcommand{\bb}{\bar{b}}
\newcommand{\cb}{\bar{c}}

\newcommand{\comm}[2]{\Delta  _{ {#1} {#2} }}

\newcommand{\commML}[3]{
{#1} \restrict{{#2}=1, {#3}=1} \cdot {#1} \restrict{{#2}=0, {#3}=0} -
{#1} \restrict{{#2}=1, {#3}=0} \cdot {#1} \restrict{{#2}=0, {#3}=1}}

\newcommand{\Spar}[3]{\frac{\partial ^{2}{#1}}{\partial x_{#2}\partial x_{#3}}}
\newcommand{\Fpar}[2]{\frac{\partial {#1}}{\partial {x_#2}}}
\newcommand{\B}[2]{B_{#1}^{#2} }
\newcommand{\dec}[3]{$(x_{#1}, x_{#2})$-decomposable mod ${#3}$}
\newcommand{\lROP}{$\bar{a}$-three-locally read-once }
\newcommand{\ignore}[1]{}

\begin{document}

\title{Characterizing Arithmetic Read-Once Formulae}

 \date{}

\author{ Ilya Volkovich \thanks{Computer Science Department and Center for Computational Intractability,
Princeton University, Princeton NJ. 
Email: {\tt ilyav@cs.princeton.edu}.
Research partially supported by NSF Award CCF 0832797.}}

\maketitle

\begin{abstract}
An \emph{arithmetic read-once formula} (ROF for short) is a
formula (i.e. a tree of computation) in which the
operations are $\{+,\times\}$ and such that every input variable
labels at most one leaf. 
We give a simple characterization of such formulae.
Other than being interesting in its own right, our characterization gives rise to a property testing algorithm for
functions computable by such formulae. 
To the best of our knowledge, prior to our work no characterization and/or property testing algorithm was known
for this kind of formulae.
\end{abstract}

\thispagestyle{empty}

\pagenumbering{arabic}

\section{Introduction}
\label{sec:Intro}

Read-once formulae (ROF) are formulae in which each variable appears at most once.
Those are the smallest possible functions that depend on all of their variables.
Although they form a very restricted model of
computation, they received a lot of attention in both the Boolean 
\cite{KLNSW93,AngluinHK93,BshoutyHH95b} and the algebraic 
\cite{HancockH91,BshoutyHH95,BshoutyBshouty98,BshoutyCleve98, ShpilkaVolkovich09, ShpilkaVolkovich14}
worlds. A polynomial $P(\bar{x})$ is a \emph{read-once polynomial} (ROP for short) 
if it can be computed by an arithmetic read-once formula.
It is not hard to see that read-once polynomials form a proper subclass of multilinear polynomials%
\footnote{A multilinear polynomial is a polynomial in which each the individual degree of each variable is at most $1$.}.

In \cite{Gurvich77} and then again in \cite{KLNSW93}, a characterization of functions
computed by Boolean read-once formulae was given. Those were referred to as ``read-once functions''.
In this work we give a characterization of functions computable by arithmetic read-once formulae.
That is, read-once polynomials.
More specifically, we prove that a polynomial $P(\bar{x})$ is a read-once polynomial if and only if
all its restrictions to three variables are read-once polynomials,
thus showing that the structural property holds globally if and only if it holds locally. 

Our structural results require that the underlying field is of polynomial size.
In case that $\size{\F}$ is too small we view the given polynomials 
as polynomials over an extension field $\Ext$ of an appropriate size.  
This is common to many structural results for polynomials (e.g. \cite{FGLSS91,AroraSudan03}).
For more details, see Section \ref{sec:lower} and discussions in \cite{KlivansSpielman01, SY10} and the references within.

Let $P \in \Fn$ be an $n$-variate polynomial over $\F$. Given an assignment $\bar{a} \in \F^n$ we say that a polynomial $P(\bar{x})$ is \emph{\lROP}\!if 
for any choice of three variables, setting the remaining variables in $P$ to $\bar{a}$ results in a read-once polynomial.
Formally, for every subset $I \subseteq [n]$ of size $\size{I}=3$ setting $x_i = a_i$ in $P$ for every $i \in [n] \setminus I$ 
results in a read-once polynomial.
Note that in terms of the restriction size our results are tight since
every bivariate multilinear polynomial is a read-once polynomial.
In other words, every multilinear polynomial  is $\bar{a}$-two-locally read-once for every $\ab \in \F^n.$
We can now give our main theorem.

\begin{THEOREM}
\label{THM:Main1}
Let $n \geq 1$ and $\F$ be a field of size $\size{\F} \geq 1.5n^3$. 
Let $P \in \Fn$ be a multilinear polynomial over $\F$. 
Then $P$ is a read-once polynomial if and only if $P$ is \lROP for each $\bar{a} \in \F^n$.
\end{THEOREM}


While establishing a structural result, iterating over all the assignments in $\F^n$ is a costly computational task. 
In order to get better algorithmic performance, we establish a more parametric version of the theorem.

\begin{THEOREM}
\label{THM:Main2}
Let $n \geq 1,\varepsilon > 0$ and $\F$ be a field of size $\size{\F} > 1.5n^4 / \varepsilon$. 
Let $P \in \Fn$ be a multilinear polynomial over $\F$.
Then $P$ is a read-once polynomial if and only if $P$ is \lROP for at least $\varepsilon$ fraction of $\bar{a} \in \F^n$.
\end{THEOREM}




We present two applications of results.
The first application is a property tester for read-once polynomials.
The construction uses our characterization and a result from \cite{FGLSS91}.
This is the first property tester for this class of polynomials. 


\begin{THEOREM}
\label{THM:Main3}
Let $n \geq 1,\delta > 0$ and $\F$ be a field of size $\size{\F} = \Omega(\frac{n^5}{\delta})$.
There exists a non-adaptive algorithm that
given oracle access to a function $f : \F^n \to \F$ runs in time $\poly(n,1/\delta)$ and outputs ``yes'' if $f$ represents a read-once polynomial.
Otherwise, if $f$ is $\delta$-far from any read-once polynomial, the algorithm outputs ``no'' with probability at least $3/4$.
\end{THEOREM}

Next, we devise algorithm for the problem of 
\emph{read-once testing}.
That is given oracle access to a polynomial $P$ decide if $P$ is a read-once 
polynomial. This problem was introduced and studied in \cite{ShpilkaVolkovich14}.

\begin{THEOREM}
\label{THM:Main4}
Let $n,d \geq 1$ and $\F$ be a field of size $\size{\F} = \Omega(n^4+d)$. 
There exists a non-adaptive algorithm that given oracle access to a polynomial $P \in \Fn$ of degree at most $d$
runs in time $\poly(n,d)$ and outputs ``yes'' if $P$ is a read-once polynomial.
Otherwise, the algorithm outputs ``no'' with probability at least $3/4$.
\end{THEOREM}

\subsection{Techniques}

We call a polynomial $P$ \emph{separable} if it
can be represented either as $P=P_{1}+P_{2}$
or as $P=P_{1} \cdot P_{2}+c$,
where $P_{1}$ and $P_{2}$ are non-constant, variable disjoint polynomials and $c$ is a field element. 
If the above does not hold, we call $P$ \emph{non-separable}.
It follows from the definition that a multivariate read-once polynomial must be separable.
Moreover, read-once polynomials can be thought of as ``strongly'' separable polynomials since 
the above holds true for $P_{1}$ and $P_{2}$ as well.

Clearly, setting some variables to field elements in a separable polynomial $P$ results in a separable polynomial.
Yet, we might get the same result even in the case that $P$ was non-separable to begin with.
In other words, setting some variables to field elements in a \textbf{non-separable} polynomial may result in a separable polynomial.
In this paper, we show how to preserve the structure of a polynomial with respect to separability.
For this purpose, we introduce the Mapping $\B{i,j}{}$ (see Section \ref{sec:B}), which is our main technical contribution.

In a nutshell, the main argument goes as follows: 
We show that for every $P$ there exists a polynomial $\Phi(P)$, related to $P$,
such that if $\bar{a}$ is a non-zero of $\Phi(P)$ fixing the variables of $P$ to $\bar{a}$ 
preserves the (non)-separable structure of $P$. This implies that $P$ is separable iff all its restrictions to $\ab$ are separable.
Consequently, if all of $P$'s restrictions are read-once polynomials then they are separable and hence $P$ itself must be separable to begin with.
Finally, we observe that if $P$ is separable with all its restrictions being read-once polynomials
then $P$ itself must be a read-once polynomial.
In terms of finding $\ab$, we note that since $\Phi(P)$ is a low-degree polynomial a typical assignment will do the job.

\subsection{Related Work}

As was mentioned earlier, arithmetic read-once formulae have received a lot of attention in literature
\cite{HancockH91,BshoutyHH95,BshoutyBshouty98,BshoutyCleve98, ShpilkaVolkovich09, ShpilkaVolkovich14}.
Several efficient reconstruction\footnote{The \emph{reconstruction} problem is defined as follows:  given an
oracle access to a read-once polynomial $P$ output a read-once formula for $P$.}
 algorithms, both deterministic \cite{ShpilkaVolkovich14}
and randomized \cite{HancockH91,BshoutyHH95,ShpilkaVolkovich14} were given. 
In particular in \cite{ShpilkaVolkovich14}, it was shown that  
a read-once polynomial $P$ can be efficiently reconstructed from its three-variate restrictions to a typical assignment.
That is, given a typical assignment $\bar{a}$,
a read-once formula for $P$ can be constructed efficiently given 
the following three-variate polynomials: 
$\set{\text{P restricted to } x_i = a_i \text{ for every } i \not \in I}_{\set{\size{I}=3}}$
\footnote{The original statement of the result of \cite{ShpilkaVolkovich14} slightly different, but they are equivalent.}.
The problem of Polynomial Identity Testing (PIT) for 
models related to read-once formulae was studied in \cite{ShpilkaVolkovich09}.
In \cite{ShpilkaVolkovich14}, the problem of \emph{read-once testing} was introduced and studied.
Formally: given oracle access to a polynomial $P$, decide if $P$ is a read-once 
polynomial\footnote{The original formulation of the problem in  \cite{ShpilkaVolkovich14} is more general.}.
It was shown the problem of read-once testing
is computationally equivalent (up to polynomial factors) to the PIT problem. 
As a corollary, 
an efficient randomized, two-sided error algorithm for the problem was obtained.
We show that our characterization can be used to devise a simpler, one-sided error algorithm for the problem.

Yet none of the previous results provide an actual characterization for arithmetic read-once formulae;
while one exists for their Boolean counterpart \cite{Gurvich77,KLNSW93}. 
In addition, unlike the results of \cite{Gurvich77,KLNSW93},
our characterization has the flavor of: ``global structure'' iff ``local structure''.
We show that no such characterization is possible for the Boolean read-once formulae, 
not even for monotone case. 
For more details, see Section \ref{sec:lower}.

\subsection{Organization}

The paper is organized as follows. In Section \ref{sec:preliminaries}, we give the
basic definitions and notations.
In Section \ref{sec:B}, we introduce the Mapping $\B{i,j}{}$ and show its main properties;
this is the main technical contribution of our paper.
Next, in Section \ref{sec:Main}, we give our main result
and prove Theorems \ref{THM:Main1} and \ref{THM:Main2}.
We present two applications of our result in Section \ref{sec:apps} proving Theorems \ref{THM:Main3}
and \ref{THM:Main4}.
We conclude the paper in Section \ref{sec:lower}
by showing some lower bounds on the required field size and discuss some impossibility results
as well as open questions.

\section{Preliminaries}
\label{sec:preliminaries}

For a positive integer $n$, we denote $[n] =\{1,\ldots,n\}$. 
Let $\F$ be a field and denote by $\overline{\F}$ its algebraic
closure \footnote{We do not assume the existence of the algebraic closure. For our purposes  $\overline{\F}$ can be replaced
by a sufficiently large extension of field of $\F$. }.
For a graph $G=(V,E)$ and a vertex $i \in V$ 
we denote by $G^i$ the graph resulting upon removing the vertex $i$ and all of its adjacent edges from $G$. 
For a polynomial $P(x_1,\ldots,x_n)$, a variable $x_i$ and a field element $\alpha$,
we denote with $P\restrict{x_i = \alpha}$ the polynomial resulting from
setting $x_i= \alpha$. Given a subset $I\subseteq [n]$ and an assignment $\bar{a}$ $\in
\F^{n}$, we define $P \restrict{\bar{x}_I = \bar{a}_I}$  to be the
polynomial resulting from setting $x_i = a_i$ for every $i \in I$.
We say that $P \in \Fn$ \emph{depends} on ${x}_{i}$ if there exist
$\bar{a}, \bar{b}\in \overline{\F}^n$ differing only on the $i$-th coordinate such that:
$P(\bar{a}) \neq P(\bar{b})$.
We denote $\var(P) \eqdef \condset{x_i}{ P \mathrm{ \; depends \; on \; } x_{i}}$.
We often denote variables
interchangeably by their index or by their label: $i$ versus $x_i$. 

\begin{definition}[Distance]
Let $f,g : \F^n \to \F$ be functions.
We define their (relative) \emph{distance} as $\Delta(f,g) \eqdef \Pr_{\ab \in \F^n}[f(\ab) \neq g(\ab)]$.
For $\delta > 0$ we say that $f$ is \emph{$\delta$-far} from $g$ if $\Delta(f,g) > \delta$.
We can extend the definition to sets of functions. Let $S$ be a non-empty set of functions.
We say that $f$ is \emph{$\delta$-far} from $S$ if $\Delta(f,g) > \delta$ for every $g \in S$.
\end{definition}

\begin{definition}[Crossing Pair]
Let $S,T$ be two non-empty sets.
We say that $(i,j)$ is a \emph{crossing pair} for $(S,T)$ if
either $i \in T \setminus S, j \in S \setminus T$ or $j \in T \setminus S, i \in S \setminus T$. 
\end{definition}

The following is a simple fact regarding two non-trivial partitions of a set.

\begin{fact}
\label{fact:pair}
Let $X$ be a set of size $\size{X} \geq 2$. Let $(T_1, X \setminus T_1)$ and
$(T_2, X \setminus T_2)$ be two non-trivial partitions of $X$.
Then there exist  $i \neq j$ such that $(i,j)$ is a crossing pair for both
$(T_1, X \setminus T_1)$ and $(T_2, X \setminus T_2)$.
\end{fact}

\begin{definition}[Decomposability]
\label{def:semi dec}
Let $P \in \Fn$ be a polynomial, $i \neq j \in [n]$ and $c \in \F$. We say that $P$ is
\emph{\dec{i}{j}{c}} if $P$ can be written 
as $P = h \cdot g + c$ where $x_i \in \var(h) \setminus \var(g)$ and $x_j \in
\var(g) \setminus \var(h)$. In other words, $(i,j)$ is a crossing pair for $(\var(g), \var(h))$.
We say that $P$ is \emph{\dec{i}{j}{\F}} if $P$ is \dec{i}{j}{c} for some $c \in \F$.
\end{definition}

\begin{definition}[Variable Separability]
\label{def:sep}
We say that a polynomial $P$ is \emph{additively separable} if $P$
can be represented as $P(\bar{x})=P_{1}(\bar{x})+P_{2}(\bar{x})$
where $P_{1}$ and $P_{2}$ are non-constant, variable disjoint polynomials. 
We say that a polynomial $P$ is \emph{multiplicatively separable} if $P$
can be represented as $P(\bar{x})=P_{1}(\bar{x})\cdot P_{2}(\bar{x})+c$
where $P_{1}$ and $P_{2}$ are non-constant, variable disjoint polynomials and
$c$ is a field element. 
We say that $P$ is \emph{separable} if it is either additively or multiplicatively separable.
\end{definition}

It is easy to see that a multilinear polynomial $P$ is multiplicatively separable
if and only if it is \dec{i}{j}{\F} for some $i$ and $j$.
We finish this part by presenting a simple result from graph theory.

\begin{lemma}
\label{lem:con}
Let $G=(V,E)$ be a graph of size at least three.
Then $G$ is connected iff there exists $k \neq \ell \in V$ such that $G^{k}$ and $G^{\ell}$ are connected.
\end{lemma}

\subsection{Partial Derivatives}
\label{sec:pd}

Partial derivatives of multilinear polynomials can be defined formally
over any field $\F$ by stipulating the partial derivative of a polynomial
over continuous domains. 

\begin{definition}[Partial Derivative]
\label{def:par} Let $P \in \Fn$ be a polynomial and let $i \in [n]$. 
We define the \emph{partial derivative of $P$ w.r.t. $x_i$} as
$\frac{\partial P}{\partial x_i} \eqdef P \restrict{x_i = 1} - P \restrict{x_i = 0}$. 
\end{definition}
\noindent Observe that for multilinear polynomials the sum, product, and chain rules carry over.  

\newpage
\subsection{Commutator}

We now formally introduce one of our main tools.
The Commutator was defined in \cite{ShpilkaVolkovich10} where it was used
for purposes of polynomial factorization. 
In \cite{ShpilkaVolkovich14}, it was used to devise new reconstruction algorithms for read-once formulae.
We recall its definition together with its main property.

\begin{definition}[Commutator]
\label{def:commutator} Let $P \in \Fn$ be a polynomial and let
$i,j \in [n]$. We define the \emph{commutator} between $x_i$ and
$x_j$ as $$\comm{i}{j} P \eqdef \commML{P}{x_i}{x_j}$$
\end{definition}

We now give the main property of the commutator. Recall Definition \ref{def:semi dec}.

\begin{lemma}[\cite{ShpilkaVolkovich10,ShpilkaVolkovich14}]
\label{lem:commutator}
Let $P \in \Fn$ be a multilinear polynomial, $i \neq j \in \var(P)$
and $c \in \F$.
Then $P$ is \dec{i}{j}{c} if and only if  $\comm{i}{j} P =  c \cdot \Spar{P}{i}{j}$.
\end{lemma}





The next useful property follows easily given the commutator.

\begin{lemma}
\label{lem:mod}
Let $P \in \Fn$ be a multilinear polynomial, $k \neq i\neq j \in [n]$ and $a_1, a_2, a_3 \in \F$ three distinct field elements.
Suppose that $P \restrict{x_k = a_t}$ is \dec{i}{j}{c} for $t \in \set{1,2,3}$ and some (fixed) $c \in \F$.
Then  $P$ is \dec{i}{j}{c}.
\end{lemma}

\begin{proof}
Consider $P' \eqdef P - c$. Then $P' \restrict{x_k = a_t}$ is \dec{i}{j}{0} for $t \in \set{1,2,3}$.
By Lemmas \ref{lem:commutator} $\comm{i}{j} P' \restrict{x_k = a_t} \equiv 0$ which implies that 
$x_k - a_t$ is a factor of $\comm{i}{j} P'$ for $t \in \set{1,2,3}$ (Lemma \ref{lem:Gauss}). As the degree of $x_k$ in $\comm{i}{j} P'$ is at most $2$
we get that $\comm{i}{j} P' \equiv 0$ and thus $P$ is \dec{i}{j}{c}.
\end{proof}

\subsection{Some Useful Facts about Polynomials}

In this section we give three facts concerning zeros of polynomials.
We begin with the Schwartz-Zippel Lemma.

\begin{lemma}[\cite{Zippel79,Schwartz80}]
\label{lem:sz}
Let $P \in \Fn$ be a non-zero polynomial of degree at most $d$ and let $V \subseteq \F$.
Then $\Pr_{\; \bar{a} \in V^n \;}[P(\bar{a}) = 0] \leq \frac{d}{\size{V}}$.
\end{lemma}

The following lemma gives a similar statement with slightly different parameterization.
A proof can be found in \cite{Alon99}.  

\begin{lemma}
\label{lem:vanish}
Let $P\in \Fn$ be a polynomial. Suppose that for every $i\in [n]$ the individual degree of each
$x_{i}$ is bounded by $d_{i}$ and let $S_{i} \subseteq \F$ be such that
$ \size{S_{i}} > d_{i}$. We denote $S=S_{1}\times S_{2}\times \cdots \times S_{n}$.
Then $P\equiv 0$ iff $P \restrict{S} \equiv 0.$
\end{lemma}

\begin{lemma}[Gauss]
\label{lem:Gauss}
Let $P\in \mathbb{F[}x_{1},x_{2},\ldots,x_{n},y]$ be a non-zero polynomial
and $g \in \Fn$ such that $P \restrict {y=g(\bar{x})} \equiv 0$ then $y-g(\bar{x})$
is an irreducible factor of $P$ in the ring $\mathbb{F[}x_{1},x_{2},\ldots,x_{n},y].$
\end{lemma}





\subsection{Read-Once Formulae and Read-Once Polynomials}

Most of the definitions or small variants of them that we give in this section, are from 
\cite{HancockH91,BshoutyHH95,ShpilkaVolkovich09,ShpilkaVolkovich14}. We
start by formally defining the notions of a read-once formula and a read-once polynomial.

\begin{definition}
\label{def:ROF,ROP} An \emph{arithmetic formula}  over a field $\F$ in the variables
$\bar{x}=(x_1,\ldots,x_n)$ is a binary tree whose leaves are labelled with input variables or field elements
and whose internal nodes (gates) are labelled with the arithmetic operations $\set{+,\times }$. 
The computation is preformed by applying the gate's operation on the incoming values.
It is easy to see that an arithmetic formula computes a polynomial. \\
In a \emph{read-once formula} (ROF for short), each input variable can label at most one leaf.
A polynomial $P(\bar{x})$ is a \emph{read-once polynomial} (ROP
for short) if it is computable by a read-once formula;
otherwise, we say that $P(\xb)$ is a \emph{read-many polynomial}.
\end{definition}

Clearly, read-once polynomials form a subclass of multilinear polynomials.
Furthermore, it is immediate from the definition that the simplest ROPs are of the form $P = \alpha \cdot x_i + \beta$
when $x_i$ is a variable and $\alpha, \beta \in \F$ are field elements. 
The following lemma, which is also immediate from the definition, provides us the structure of more complex ROPs. 

\begin{lemma}[ROP Structural Lemma]
\label{ROF_Representation_Lemma} 
A polynomial $P$ with $\size{\var(P)} \geq 2$ is a ROP iff
it can be presented in one of
the following forms:
\begin{enumerate}
\item $P(\bar{x})=P_{1}(\bar{x})+P_{2}(\bar{x})$

\item $P(\bar{x})=P_{1}(\bar{x})\cdot P_{2}(\bar{x})+c$
\end{enumerate}
where $P_{1}$ and $P_{2}$ are non-constant, variable disjoint ROPs and $c \in \F$
is a field element.
\end{lemma}

In terms of Definition \ref{def:sep} we get that each ROP with at least two variables is separable. 
On the other hand, observe that each bivariate multilinear polynomial is separable and thus is read-once.
Moreover, each trivariate multilinear polynomial is read-once iff it is separable.

We now define the important notion of the
\emph{gate-graph} of a polynomial. A similar notion was
defined in \cite{HancockH91,BshoutyHH95,ShpilkaVolkovich14} where it was used as a core tool in ROF reconstruction algorithms. 
Here we define it in a slightly more general form:


\begin{definition}[Gate Graph]
\label{def:gate graph}
Let $P \in \Fn$ be a polynomial. 
The \emph{gate graph} of $P$, denoted by $G_{P} = (V_P, E_P)$,
is an undirected graph whose vertex set is $V_P =\var(P)$ and its edges are defined as follows: 
$(i,j) \in E_P$ if and only if $\Spar{P}{i}{j} \nequiv 0$.
\end{definition}

In \cite{ShpilkaVolkovich14}, it was observed that given a ROP $P$ and $x_i \neq x_j \in \var(P)$ 
we have that $\Spar{P}{i}{j} \nequiv 0$ iff in every ROF computing $P$ the arithmetic operation that labels 
the least common ancestor of the unique input nodes of $x_i$ and $x_j$ is $\times$.  
We can extend this observation further.

\begin{observation}
\label{obs:plus}
Let $n \geq 2$ and let $P \in \Fn$ be a multilinear polynomial. 
Then $P$ is additively separable iff $G_P$ is disconnected. 
\end{observation}







\section{The Mapping $\B{i,j}{}$}
\label{sec:B}

In this section we present our main tool along with its properties.
This is the main technical contribution of the paper.

\begin{definition}
\label{def:Bij}
For $i \neq j \in [n]$  let $\B{i,j}{}:\Fn \to \Fnxy$ be a polynomial mapping defined as follows:
\begin{equation*}
\B{i,j}{}(P)(\bar{x},\bar{y}) \eqdef
\left| \left(
\begin{array}{cc}
  \comm{i}{j}(P)(\bar{x}) &	\comm{i}{j}(P)(\bar{y}) \\
  \Spar{P}{i}{j}(\bar{x}) & \Spar{P}{i}{j}(\bar{y}) \\
\end{array}
\right)\right|
=  \comm{i}{j}(P)(\bar{x}) \cdot  \Spar{P}{i}{j}(\bar{y}) -  \Spar{P}{i}{j}(\bar{x}) \cdot 	\comm{i}{j}(P)(\bar{y})
\end{equation*}
For $i \neq j \in [n]$ and $I \subseteq [n] \setminus \set{i,j}$
let $\B{i,j}{I}:\Fn \to \Fnxy$ be a polynomial mapping defined as 
$\B{i,j}{I}(P)(\bar{x},\bar{y}) \eqdef \B{i,j}{}(P) \restrict{\bar{y}_I = \bar{x}_I}$.
\end{definition}
Intuitively, the purpose of $\B{i,j}{}(P)$ is to preserve the structure of $P$ w.r.t multiplicative separability.
The following lemma lists several useful
properties of $\B{i,j}{}$ that shed light on this intuition. We will use them implicitly in our proofs.

\begin{lemma}
\label{lem:Bij prop}
Let $n \geq 4$ and let $P \in \Fn$ be a multilinear polynomial. Let $i,j,k,\ell \in [n]$
be distinct indices. Then
the following properties hold:

\begin{enumerate}
\item $\B{i,j}{}(P) \equiv 0$ iff either $\Spar{P}{i}{j} \equiv 0$ or $P$ is \dec{i}{j}{\F}.

\item Let $\alpha \in \F$. Then $\B{i,j}{}(P \restrict{x_k = \alpha}) =  \B{i,j}{\set{k}}(P) \restrict{x_k = \alpha}$.

\item If $\B{i,j}{\set{k}}(P) \equiv 0$ and $\B{i,j}{\set{\ell}}(P) \equiv 0$ then $\B{i,j}{}(P) \equiv 0$.

\item Let $I \subseteq [n] \setminus \set{i,j}$ be of size $\size{I} \leq n-3$.
 If $\B{i,j}{I}(P) \nequiv 0$ then there exists $I \subseteq J \subseteq [n] \setminus \set{i,j}$ of size $\size{J} = n-3$ 
such that  $\B{i,j}{J}(P) \nequiv 0$.
\end{enumerate}
\end{lemma} 

\begin{proof} 
$\;$
\begin{enumerate}
\item Suppose $\B{i,j}{}(P) \equiv 0$.
In other words:
$\comm{i}{j}(P)(\bar{x}) \cdot  \Spar{P}{i}{j}(\bar{y}) =  \Spar{P}{i}{j}(\bar{x}) \cdot 	\comm{i}{j}(P)(\bar{y})$.
If $\Spar{P}{i}{j} \equiv 0$ we are done. Otherwise, we can write: 
$$\frac{\comm{i}{j}(P)}{\Spar{P}{i}{j}}(\bar{x}) =  \frac{\comm{i}{j}(P)}{\Spar{P}{i}{j}}(\bar{y}).$$
As the LHS and the RHS defined on disjoint sets of variables, it must be the case that  
$$\frac{\comm{i}{j}(P)}{\Spar{P}{i}{j}}(\bar{x}) =  c = \frac{\comm{i}{j}(P)}{\Spar{P}{i}{j}}(\bar{y})$$
for some $c \in \F$ and the conclusion follows from Lemma \ref{lem:commutator}.

\item Follows from the definition.

\item  
Suppose $\B{i,j}{\set{k}}(P) \equiv 0$ and $\B{i,j}{\set{\ell}}(P) \equiv 0$. Assume WLOG that $\Spar{P}{i}{j} \nequiv 0$ (otherwise we are done).
Let $a,b \in \overline{\F}$ be such that $\Spar{P}{i}{j} \restrict{x_k=a,x_\ell=b} \nequiv 0$.
We have that $\B{i,j}{}(P \restrict{x_k = a}) =  \B{i,j}{\set{k}}(P) \restrict{x_k = a} \equiv 0$ and similarly $\B{i,j}{}(P \restrict{x_\ell = b}) \equiv 0$.
For the sake of simplicity, assume WLOG that $k=1$ and $\ell = 2$. 
Hence, there exist $c_1, c_2 \in \F$ such that:
$$P(\bar{x}) \restrict{x_1 = a} = h_1(\bar{x}_{L_1}) \cdot g_1(\bar{x}_{R_1}) + c_1$$
$$P(\bar{x}) \restrict{x_2 = b} = h_2(\bar{x}_{L_2}) \cdot g_2(\bar{x}_{R_2}) + c_2$$
where $2 \in L_1$, $1 \in L_2$ and $(i,j)$ is a crossing pair for both $(L_1, R_1)$ and $(L_2, R_2)$. 
Implying:
\begin{equation*}
h_1(\bar{x}_{L_1}) \restrict{x_2=b} \cdot g_1(\bar{x}_{R_1}) + c_1 - c_2 = P(\bar{x}) \restrict{x_1 = a, x_2=b} - c_2
= h_2(\bar{x}_{L_2}) \restrict{x_1=a} \cdot g_2(\bar{x}_{R_2}).
\end{equation*}
By applying $\comm{i}{j}$ and $\Spar{}{i}{j}$ to the equation we obtain:
\begin{gather*}
(c_1 - c_2) \cdot \Spar{P}{i}{j} \restrict{x_1=a,x_2=b} = 
(c_1 - c_2) \cdot \Spar{ \left(h_1 \restrict{x_2=b} \cdot g_1 \right) }{i}{j} \\ =
\comm{i}{j}\left(h_1 \restrict{x_2=b} \cdot g_1 + c_1 - c_2 \right) = 
\comm{i}{j}\left(h_2 \restrict{x_1=a} \cdot g_2 \right) \equiv 0.
\end{gather*}
Since $a$ and $b$ were chosen such that $\Spar{P}{i}{j} \restrict{x_1=a,x_2=b} \nequiv 0$ we obtain that $c_1 = c_2$.
By repeating this reasoning, we can fix $b$ and choose many distinct elements $a_t \in \F$ for which 
$\Spar{P}{i}{j} \restrict{x_1=a_t,x_2=b} \nequiv 0$ would imply that 
$P \restrict{x_1 = a_t}$ is \dec{i}{j}{c_1}. 
By Lemma \ref{lem:mod}, $P$ is \dec{i}{j}{c_1} and thus $\B{i,j}{}(P) \equiv 0$.

\item 
Follows by an iterative application of the following claim, a generalization of Property $3$.

\begin{claim}
Let $i,j,k,\ell \in [n]$ be distinct indices and let $I \subseteq [n] \setminus \set{i,j,k,\ell}$.
Then $\B{i,j}{I \cup \set{k}}(P) \equiv 0$ and $\B{i,j}{I \cup \set{\ell}}(P) \equiv 0 \implies \B{i,j}{I}(P) \equiv 0$.
\end{claim}

\begin{proof}[Proof of the claim]
Assume for a contradiction that $\B{i,j}{I}(P) \nequiv 0$. Then there exists $\bar{a} \in \overline{\F}$ such that 
$\B{i,j}{I}(P) \restrict{\bar{x}_I = \bar{a}_I} \nequiv 0$ or equivalently 
$\B{i,j}{}(P \restrict{\bar{x}_I = \bar{a}_I}) \nequiv 0$. Applying Property $3$ that we have just proved, we obtain WLOG that
$\B{i,j}{I \cup \set{k}}(P) \restrict{\bar{x}_I = \bar{a}_I} = \B{i,j}{\set{k}}(P \restrict{\bar{x}_I = \bar{a}_I}) \nequiv 0$
leading us to a contradiction.
\end{proof}
\unskip
\end{enumerate}
\end{proof}

To provide some addition intuition on the defined mapping
we present a useful application. 
We exhibit a simple criterion (to be used later) which can be applied to test if a given trivariate multilinear polynomial is a read-once polynomial.

\begin{lemma}
\label{lem:3ROP}
Let $P(x_1, x_2, x_3) \in \F[x_1,x_2,x_3]$ be a trivariate multilinear polynomial.
Then $P$ is a ROP iff at least two of the following polynomials are identically zero $\set{\B{1,2}{}(P) \;,\; \B{1,3}{}(P) \;,\; \B{2,3}{}(P)}$.
\end{lemma}

\begin{proof}
Assume WLOG that $\var(P) = [3]$. Otherwise $P$ is a uni/bivariate polynomial and is clearly a ROP.
In addition, from Lemma \ref{ROF_Representation_Lemma} $P$ is a ROP iff it is separable. So, we show the claim regarding separability. 
Suppose $P$ is separable. WLOG either $P(x_1, x_2, x_3) = P_1(x_1) + P(x_2, x_3)$ or $P(x_1,x_2, x_3) = P_1(x_1) \cdot P(x_2, x_3) + c$.
In both cases, $\B{1,2}{}(P)= \B{1,3}{}(P) \equiv 0$. 
Now assume WLOG that $\B{1,2}{}(P)= \B{1,3}{}(P) \equiv 0$.
By Lemma \ref{lem:Bij prop}, either $\Spar{P}{1}{2} \equiv 0$ or $P$ is \dec{1}{2}{\F}.
If the latter holds, then $P$ is multiplicatively separable and we are done.
By the same reasoning, we can assume WLOG that $P$ is not \dec{1}{3}{\F} either.
Consequently $\Spar{P}{1}{2} = \Spar{P}{1}{3} \equiv 0$, which implies
that $G_P$ - the gate graph of $P$ is disconnected.
By Observation \ref{obs:plus}, $P$ must be additively separable which completes the proof.
\end{proof}

\section{Main}
\label{sec:Main}

In this section we give our main results proving Theorems \ref{THM:Main1} and \ref{THM:Main2}. 
As was suggested earlier, we would like to preserve the structure of a given polynomial $P$ w.r.t separability.
To this end, we define the following polynomial mapping:

\begin{definition}
\label{def:phi}
$\phi:\Fn \to \Fnxy$ \\
$$\phi(P)(\bar{x},\bar{y}) \eqdef 
\prod \limits _{t \in [n]} \Fpar{P}{t}(\bar{x}) \cdot
\prod \limits _{i \neq j \in [n]} \Spar{P}{i}{j}(\bar{x}) \cdot 
\prod \limits _{k \neq i,j}  \B{i,j}{\set{k}}(P)(\bar{x},\bar{y})$$
when the product is only on the corresponding \textbf{non-zero} multiplicands.
If all the corresponding multiplicands are identically zero, we define $\phi(P)(\bar{x},\bar{y}) \eqdef 1$.
\end{definition}

The next propositions demonstrate the crucial properties of $\phi(P)$. 
In what follows, 
let $P \in \Fn$ be a multilinear polynomial and $\bar{a} \in \F^n$ be such that $\phi(P)(\bar{a},\bar{y}) \nequiv 0 $.  

\begin{proposition}
\label{prop:Gp}
For each $k \in [n]$, it holds that $G_{P \restrict{x_k = a_k}} = G^{k}_P$.
\end{proposition}

\begin{proof}
Clearly, $G_{P \restrict{x_k = a_k}} \subseteq G^{k}_P$. Now, let $(i,j) \in  G^{k}_P$.
By definition, $\Spar{P}{i}{j} \nequiv 0$ and hence $\Spar{P}{i}{j}$ appears as a multiplicand in $\phi(P)$.
Since $\phi(P)(\bar{a},\bar{y}) \nequiv 0 $ we have that
$\Spar{P \restrict{x_k = a_k}}{i}{j} \nequiv 0$ implying that $(i,j) \in G_{P \restrict{x_k = a_k}}$
and thus establishing that $ G^{k}_P \subseteq G_{P \restrict{x_k = a_k}}$.
\end{proof}

\begin{proposition}
\label{prop:B}
Let $k \neq \ell \in [n]$. 
If $\B{i,j}{}(P \restrict{x_u = a_u}) \equiv 0$ for $u=k,\ell$ then $\B{i,j}{}(P) \equiv 0$.
\end{proposition}

\begin{proof}
Assume for a contradiction that $\B{i,j}{}(P) \nequiv 0$.
Then $\B{i,j}{\set{u}}(P) \nequiv 0$ for either $u=k$ or $u=\ell$.
Suppose, $u=k$. Then $\B{i,j}{\set{k}}(P)$ appears as a multiplicand in $\phi(P)$.
Since $\phi(P)(\bar{a},\bar{y}) \nequiv 0 $ we have that $\B{i,j}{}(P \restrict{x_k = a_k}) = \B{i,j}{\set{k}}(P) \restrict{x_k = a_k} \nequiv 0$,
thus leading to a contradiction. 
\end{proof}

\begin{proposition}
\label{prop:Sep ROP}
$P$ is a ROP iff $P$ is separable and for each $k \in [n]$ 
$P \restrict{x_k = a_k}$ is a ROP.
\end{proposition}

\begin{proof} 
There are two case to consider: \\
\textbf{Case $1$:} $P(\xb_L, \xb_R) = P_1(\xb_L)  + P_2(\xb_R)$. Pick $k \in R$.
Then $P \restrict{x_k = a_k} = P_1(\xb_L)  + P_2(\xb_R)\restrict{x_k = a_k}$ is a ROP,
implying that $P_1(\xb_L) = P(\xb_L, \ab_R) - P_2(\ab_R)$ is a ROP as well.
Similarly, we get that  $P_2(\xb_R)$ is a ROP. As $P_1$ and $P_2$ are defined over disjoint sets of variables
by Lemma \ref{ROF_Representation_Lemma}, $P$ is a ROP to begin with.
\\ \\ \textbf{Case $2$:} $P(\xb_L, \xb_R) = P_1(\xb_L) \cdot P_2(\xb_R) + c$. 
Pick $t \in L$. Then $\Fpar{P_1(\ab_L)}{t} \cdot P_2(\ab_R) = \Fpar{P(\ab_L, \ab_R)}{t} \neq 0$ and in particular $P_2(\ab_R) \neq 0$.
Now pick $k \in R$.
Then $P \restrict{x_k = a_k} = P_1(\xb_L) \cdot P_2(\xb_R)\restrict{x_k = a_k} + c$ is a ROP,
implying that $P_1(\xb_L) = \frac{P(\xb_L, \ab_R) - c}{P_2(\ab_R)}$ is a ROP as well. Note that 
the operation is well-defined as $P_2(\ab_R) \neq 0$.
Similarly, we get that  $P_2(\xb_R)$ is a ROP. As $P_1$ and $P_2$ are defined over disjoint sets of variables
by Lemma \ref{ROF_Representation_Lemma} $P$ is a ROP to begin with.
\end{proof}

We can now prove our main result.
We start by proving a weaker, ``baby'' case of the result, 
which will be used as an inductive step in the main proof. 

\newpage

\begin{lemma}
\label{lem:main}
Let $n \geq 4$ and let $P \in \Fn$ be a multilinear polynomial.
Let $\bar{a} \in \F^n$ be such that $\phi(P)(\bar{a},\bar{y}) \nequiv 0 $.  
Then $P$ is a ROP iff for each $k \in [n]$ $P \restrict{x_k = a_k}$ is a ROP. 
\end{lemma}

\begin{proof}
First of all, we can assume WLOG that $\var(P) = [n]$. Otherwise, let $k  \in [n] \setminus \var(P)$.
Then $P = P \restrict{x_k = a_k}$ and we are done. 
Given Proposition \ref{prop:Sep ROP}, it is sufficient to show that $P$ is separable. 
Consider the graphs $\set{G^k_P}_{k \in [n]}$. There can be two cases: \\
\textbf{Case $1$:} There exists at most one $k \in [n]$ such that $G^k_P$ is connected.
In this case, by Lemma \ref{lem:con} $G_P$ is disconnected and hence $P$ is additively separable from Observation \ref{obs:plus}. 
\\ \\ \textbf{Case $2$:} There exist $k \neq \ell \in [n]$ such that $G^k_P$ and $G^\ell_P$ are both connected.
We claim that in this case $P$ is multiplicatively separable. For the sake of simplicity, assume WLOG that $k=1$ and $\ell = 2$. 
By Proposition \ref{prop:Gp} and Observation \ref{obs:plus},
it must be the case that both $P \restrict{x_1 = a_1}$ and $P \restrict{x_2 = a_2}$ are multiplicatively separable.
Moreover, by the properties of $\bar{a}$, $\var(P \restrict{x_u = a_u}) = [n] \setminus \set{u}$, for $u=1,2$.
Hence, there exist $c_1, c_2 \in \F$ such that: 
$$P(\bar{x}) \restrict{x_1 = a_1} = h_1(\bar{x}_{L_1}) \cdot g_1(\bar{x}_{R_1}) + c_1$$
$$P(\bar{x}) \restrict{x_2 = a_2} = h_2(\bar{x}_{L_2}) \cdot g_2(\bar{x}_{R_2}) + c_2$$
where $u \in R_{3-u}$ and  $L_u \dot \cup R_u = [n] \setminus \set{u}$ for $u = 1,2$. 
We consider three sub-cases:
\\ \\ \textbf{Case $2a$:} $\size{R_1}, \size{R_2} \geq 2$. In this case we have that 
$(L_1, R_1 \setminus \set{2})$ and $(L_2, R_2 \setminus \set{1})$ are both non-trivial partitions of the set $[n] \setminus \set{1,2}$.
As $n \geq 4$, by Fact \ref{fact:pair} there exist $i \neq j $ such that $(i,j)$ is a crossing pair for both 
$(L_1, R_1 \setminus \set{2})$ and $(L_2, R_2 \setminus \set{1})$,
and hence for both $(L_1, R_1)$ and $(L_2, R_2)$.
In other words, both $P(\bar{x}) \restrict{x_1 = a_1}$  and $P(\bar{x}) \restrict{x_2 = a_2}$ are 
\dec{i}{j}{\F}
and consequently $\B{i,j}{\set{u}}(P) \restrict{x_u = a_u} = \B{i,j}{}(P \restrict{x_u = a_u}) \equiv 0$ for $u=1,2$.
By Proposition \ref{prop:B}, we get that $\B{i,j}{}(P) \equiv 0$.  As  $\Spar{P}{i}{j} \nequiv 0$ we conclude that 
$P$ is \dec{i}{j}{\F} 
and thus is multiplicatively separable.
\\ \\ \textbf{Case $2b$:} $\size{R_1} = 1$ (i.e. $R_1 = \set{2}$). We show that this sub-case reduces to the previous sub-case.
We have that $L_1 = [n] \setminus \set{1,2}$ and thus
$(2,i) \in G^{1}_P \subseteq G_P$ for $3 \leq i \leq n$. Now, pick $u \in L_2$ and $w \not \in \set{1,2,u}$. 
Recall that $1 \in R_2$ and note that $u \not \in \set{1,2}$. As $n \geq 4$, such values always exist. Moreover, we can assume WLOG that $w=3$ and $u=4$.   
Observe that $(1,4) \in G^{2}_P \subseteq G_P$, implying that $G^{3}_P$ is a connected graph. Repeating the reasoning of Case $2$ we can write:
$$P(\bar{x}) \restrict{x_3 = a_3} = h_3(\bar{x}_{L_3}) \cdot g_3(\bar{x}_{R_3}) + c_3$$
where $2 \in R_{3}$, $L_3 \dot \cup R_3 = [n] \setminus \set{3}$ and $c_3 \in \F$.
Now, if there exists $i \in L_3$ such that $4 \leq i \leq n$, then $(2,i)$ is a crossing pair for both $(L_1, R_1)$ and $(L_3, R_3)$.  
Otherwise, $R_3 = \set{2} \cup \condset{i}{4 \leq i \leq n}$ implying that $L_3 = \set{1}$ and hence $(1,4)$
is a crossing pair for both $(L_2, R_2)$ and $(L_3, R_3)$. 
Both outcomes reduce to Case $2a$.
\\ \\ \textbf{Case $2c$:} $\size{R_2} = 1$. Similar to Case $2b$.
\end{proof}

We now move to the proof of the main result.
As was suggested earlier, we would like to apply induction. In order to use induction, we need to ensure
that the crucial properties of $\phi(P)$ (i.e. Proposition \ref{prop:Gp}, \ref{prop:B} and \ref{prop:Sep ROP})
carry over throughout the inductive steps. To this end, we define the following``induction friendly'' version of $\phi(P)$.

\begin{definition}
\label{def:Phi}
$\Phi :\Fn \to \Fnxy$ \\
$$\Phi(P)(\bar{x},\bar{y}) \eqdef 
\prod \limits _{t \in [n]} \Fpar{P}{t}(\bar{x}) \cdot
\prod \limits _{i \neq j \in [n]} \Spar{P}{i}{j}(\bar{x}) \cdot 
\prod \limits _{J \subseteq [n] \setminus \set{i,j}, \size{J} = n-3}  \B{i,j}{J}(P)(\bar{x},\bar{y})$$
when the product is only on the corresponding \textbf{non-zero} multiplicands.
If all the corresponding multiplicands are identically zero we define $\Phi(P)(\bar{x},\bar{y}) \eqdef 1$.
\end{definition}

The following proposition shows that $\Phi(P)$ is indeed an ``induction friendly'' version of  $\phi(P)$
where each inductive step is reflected by fixing one variable at a time until we are left with three variables only. 

\begin{proposition}
\label{prop:Phi}
Let $P \in \Fn$ be a multilinear polynomial and let $\bar{a} \in \F^n$ be such that $\Phi(P)(\bar{a},\bar{y}) \nequiv 0 $.  
Let $I \subseteq [n]$ be of size  $\size{I} \leq n-4$. Then $\phi(P \restrict{\bar{x}_I = \bar{a}_I})(\bar{a},\bar{y}) \nequiv 0$.
\end{proposition}

\begin{proof}
$\phi$ contains three types of non-zero multiplicands.
We show that $\bar{a}$ is their common non-zero. 
Suppose that $\Spar{P \restrict{\bar{x}_I = \bar{a}_I}}{i}{j} \nequiv 0$. 
Then in particular  $\Spar{P}{i}{j} \nequiv 0$. By the definition of $\Phi$ we have that $\Spar{P}{i}{j}(\bar{a}) \neq 0$,
implying that  $\Spar{P \restrict{\bar{x}_I = \bar{a}_I}}{i}{j}(\bar{a}) \neq 0$.
Similar reasoning works for $\Fpar{P}{t} \nequiv 0$.
Now, suppose that $\B{i,j}{\set{k}}(P  \restrict{\bar{x}_I = \bar{a}_I}) \nequiv 0$,
which is equivalent to $\B{i,j}{I \cup \set{k}}(P) \restrict{\bar{x}_I = \bar{a}_I} \nequiv 0$
and hence implies $\B{i,j}{I \cup \set{k}}(P) \nequiv 0$.
As $\size{I \cup \set{k}} \leq n-3$ by Lemma \ref{lem:Bij prop} 
there exists $(I \cup \set{k}) \subseteq J \subseteq [n] \setminus \set{i,j}$ of size $\size{J} = n-3$ 
such that  $\B{i,j}{J}(P) \nequiv 0$.
By the definition of $\Phi$ we have that $\B{i,j}{J}(P)(\bar{a}) \neq 0$.
As $\B{i,j}{J}(P)$ is a restriction of $\B{i,j}{I \cup \set{k}}(P)$ we get that $\B{i,j}{I \cup \set{k}}(P)(\bar{a}) \neq 0$
as required.
\end{proof}

We can finally state our main theorem from which Theorems \ref{THM:Main1} and \ref{THM:Main2} follow as corollaries.

\begin{theorem}[Main]
\label{thm:Main Tech}
Let $P \in \Fn$ be a multilinear polynomial and let $\bar{a} \in \F^n$ be such that $\Phi(P)(\bar{a},\bar{y}) \nequiv 0 $.  
Then $P$ is a ROP iff for each $I \subseteq [n]$ of size $\size{I} = 3$ 
$P \restrict{\bar{x}_{[n] \setminus I} = \bar{a}_{[n] \setminus I}}$ is a ROP.
\end{theorem}

\begin{proof}
For $S \subseteq [n]$, we define $Q_S \eqdef P \restrict{\bar{x}_{[n] \setminus S} = \bar{a}_{[n] \setminus S}}$.
We prove that $Q_S$ is a ROP when $\size{S} \geq 3$ by induction on $\size{S}$.
The base case $\size{S} = 3$ corresponds to the conditions of the theorem. Now suppose that $\size{S} \geq 4$.
Pick $k \in S$. We have that $Q_{S} \restrict{x_k = a_k} = Q_{S \setminus \set{k}}$ and
thus $Q_{S} \restrict{x_k = a_k}$ is a ROP by the induction hypothesis.
By Proposition \ref{prop:Phi}, $\phi(Q_S)(\bar{a},\bar{y}) \nequiv 0$. 
Given this, Lemma \ref{lem:main} implies that $Q_{S}$ is a ROP to begin with.
To finish the proof, observe that $Q_{[n]} = P$.
\end{proof}

We now turn to the proofs of Theorems \ref{THM:Main1} and \ref{THM:Main2}. 

\begin{proof}[Proof of Theorem \ref{THM:Main1}]
The first direction is trivial. A restriction of a read-once polynomial is itself a read-once polynomial.
For the other direction, observe that the individual degree of each $x_i$ in $\Phi(P)(\bar{x},\bar{y})$ is less than $1.5n^3$.
As $\Phi(P)(\bar{x},\bar{y}) \nequiv 0$ by Lemma \ref{lem:vanish}, there exists $\bar{a} \in \F^n$ such that
$\Phi(P)(\bar{a},\bar{y}) \nequiv 0$. By the main theorem, $P$ is a read-once polynomial. 
\end{proof}

\begin{proof}[Proof of Theorem \ref{THM:Main2}] 
Let us view $\Phi(P)(\bar{x},\bar{y})$ as a polynomial over ${\mathbb{F}(y_1, y_2, \ldots, y_n)[x_1,x_2,\ldots ,x_{n} ]}$.
Given this and by the Schwartz-Zippel Lemma (Lemma \ref{lem:sz}), 
$\Pr_{\; \bar{a} \in \F^n \;}[\Phi(P)(\bar{a},\bar{y}) \equiv 0] \leq \frac{1.5n^4}{\size{\F}} < \varepsilon $ which implies 
that there exists $\bar{a} \in \F^n$ such that $\Phi(P)(\bar{a},\bar{y}) \nequiv 0$ and we are done.
\end{proof}


\section{Applications}
\label{sec:apps}

In this section we give two applications of our results.
The first application is a property testing algorithm for read-once polynomials.
The second application is an efficient algorithm  for the \emph{read-once
testing problem} (see below). The key difference between the problems is that in the first
case we need to test whether or not a given function is close (in the Hamming distance) to a function representable by a read-once polynomial.
While in second case, we need to determine whether a given polynomial equals to a read-once polynomial as a formal sum of monomials.
For example, the $x^2-x$ represent a function computable by a read-once polynomial over the field with two elements,
while from the formal point of view, its not even a multilinear polynomial.
We note that for polynomials over sufficiently large fields there is no difference between the functional and the formal equalities.

\subsection{Property Testing for Read-Once Polynomials}

A property tester for a property $\Prop$ is a procedure that given oracle access to a function $f : \F^n \to \F$ 
tests if $f$ represents a function from $\Prop$ or $f$ is ``far'' for any such function.
In this section we construct a property tester for read-once polynomials thus proving Theorem \ref{THM:Main3}.
We build on the property tester for multilinear polynomials of Feige et al. \cite{FGLSS91}.
The following definitions are from \cite{FGLSS91} or slight modifications of them:

\begin{definition}[Aligned Triples]
We call a set of three distinct points $\set{\bar{\alpha},\bar{\beta},\bar{\gamma}} \subseteq \F^n$ 
an \emph{aligned triple} if there exists a coordinate $i \in [n]$ such that
they differ only on the $i$-th coordinate.
Let $f : \F^n \to \F$ be a function.
Define $\tilde{f}(x_i) \eqdef f \restrict{\xb_{[n] \setminus \set{i}} \; = \; \bar{\alpha}_{[n] \setminus \set{i}}}$.
We say that the aligned triple $\set{\bar{\alpha},\bar{\beta},\bar{\gamma}}$
is \emph{$f$-linear} if the univariate interpolating polynomial of $\tilde{f}(x_i)$ over the set $\set{\alpha_i, \beta_i, \gamma_i}$
is of a degree at most $1$ in $x_i$.
Finally, we denote by $\tau(f)$ the probability that a random aligned triple is not $f$-linear. Formally: 

\begin{equation*}
\tau(f) \eqdef \Pr_{\text{aligned triple } \set{\bar{\alpha},\bar{\beta},\bar{\gamma}} \subseteq \F^n}
\left[ \set{\bar{\alpha},\bar{\beta},\bar{\gamma}} \text{ is not $f$-linear} \right].
\end{equation*}
\end{definition}

Given this terminology, we can now state the result of Feige et al. that gives rise to a property tester for multilinear polynomials.

\begin{lemma}[Theorem 9 of \cite{FGLSS91} reformulated]
\label{lem:ML test}
Let $n \geq 1, \delta > 0$ and $\F$ be a field of size $\size{\F} > 12n / \delta + 2$.
Let $f : \F^n \to \F$ be an arbitrary function. If $f$ is $\delta$-far from any multilinear polynomial over $\F$ 
then $\tau(f) \geq \delta/30n$.
\end{lemma}

In other words, it is sufficient to test multilinearity for random triples of points differing only on one coordinate.
We show that for the case of read-once polynomials, it is sufficient to test the property 
for random triples of points differing only on \emph{three} coordinates. 

\begin{algorithm} 
\label{alg:property ROT}
    \KwIn{$n \geq 1, \delta > 0$, oracle access to $f : \F^n \to \F$.}
    \KwOut{``yes'' if $f$ represents a ROP, ``no'' if $f$ is $\delta$-far from any ROP.} 

    Pick $\ab, \bb, \cb \in \F^n$ at random without repetitions (that is, $a_i \neq b_i \neq c_i$) \;
	 \ForEach{$I \subseteq [n]$ of size $\size{I} = 3$} {
        Set $\tilde{f}(\xb_{I}) \eqdef f \restrict{\xb_{[n] \setminus I} \; = \; \ab_{[n] \setminus I}}$ \;
        Set $S_I \eqdef \prod _{i \in I}  \set{a_i, b_i, c_i}$ ($3 \times 3$ Cartesian product) \;
        Compute $\tilde{P}(\xb_{I})$ - the three-variate interpolating polynomial of $\tilde{f}(\xb_{I})$ over the set $S_I$ \;  
   	\If{$\tilde{P}$ is not a multilinear polynomial}
    {Output ``no'' \;} 
     \Else{Check if $\tilde{P}$ is a ROP using Lemma \ref{lem:3ROP} \; }
	 \label{Lin10}
	}
    Output ``yes'' iff $f$ passes all the tests \;

\caption{Property Tester for Read-Once Polynomials}
\end{algorithm}

\begin{lemma}
\label{lem:property ROT}
Let $n \geq 1,\delta > 0$ and $\F$ be a field of size $\size{\F} > 24 \cdot \max \set{\frac{n}{\delta},n^5}$.
Given oracle access to a function, $f : \F^n \to \F$
Algorithm \ref{alg:property ROT} runs in time $\poly(n,1/\delta)$ and outputs ``yes'' if $f$ represents a read-once polynomial.
Otherwise, if $f$ is $\delta$-far from any read-once polynomial, the algorithm outputs ``no'' with probability at least 
$1 - \exp(-\delta-\frac{1}{n^4})$.
\end{lemma}

\begin{proof}
The claim regarding the running time is immediate from the description of the algorithm.
Observe that the test in Line \ref{Lin10} is actually an identity test for a quadratic univariate polynomial.
For the correctness, clearly, if $f$ represents a ROP then it passes all the tests.
Suppose that $f$ is $\delta$-far from any ROP. 
Set $\delta' \eqdef \min \set{ \frac{\delta}{2} , \frac{1}{n^4}}$. We divide our analysis into two cases: \\ \\
\textbf{Case $1$:} $f$ is $\delta'$-far from any multilinear polynomial over $\F$.
Observe that if $f$ passes all the tests, then the algorithm encounters at least $n$ random $f$-linear aligned triples.
By Lemma \ref{lem:ML test} the probability of the event is at most $(1 - \delta'/30n)^{n} \leq \exp(-\delta') \leq \exp(-\delta - \frac{1}{n^4})$.
\\ \\ \textbf{Case $2$:} There exists a multilinear polynomial $P \in \Fn$ such that $\Delta(f,P) \leq \delta' \leq \delta/2$. 
Since $f$ is $\delta$-far from any ROP, $P$ is not a ROP.
We claim that this case affectively reduces to Theorem \ref{THM:Main2}. Intuitively, if we executed the algorithm on $P$ instead
of $f$, Theorem \ref{THM:Main2} would guarantee small failure probability. On the other hand, $f$ and $P$ are very close
and we only query $f$ on a small set of random points, so with high probability we will actually see the values of $P$.
Formally, consider a single iteration $k$. Let us denote by $E_{pass}$ the event that $f$ passes all the tests
and by $E_{eq}$ the event that $f$ and $P$ are equal on all the $\BigO(n^3)$ query points in this iteration. 
We have that:
\begin{equation*}
\Pr[E_{pass}] \leq \Pr[E_{pass} \; | \; E_{eq}] + \Pr[\bar{E}_{eq}] \leq \BigO(\frac{1}{n}) + \BigO(\frac{1}{n}) = \BigO(\frac{1}{n})
\leq \exp(-\delta-\frac{1}{n^4}).
\end{equation*}
The upper bounds on the terms 
follow from Theorem \ref{THM:Main2} and the fact that $\Delta(f,P) \leq \frac{1}{n^4}$, respectively.
\end{proof}

Theorem \ref{THM:Main3} follows as a corollary of the lemma by repeating
the algorithm $\BigO(\frac{1}{\delta+1/n^4})$ times.

\subsection{Read-Once Testing}

The second application is for read-once testing. This problem was first defined and studied in  \cite{ShpilkaVolkovich14}.

\begin{problem}[Problem $1.1$ in \cite{ShpilkaVolkovich14}]
\label{prob:int:REOFT} Given oracle access to a polynomial $P$, decide if $P$ is a read-once 
polynomial, and if the answer is positive output a read-once formula for it.
\end{problem}

The original formulation of the problem is actually more general.
Here we focus on randomized algorithms for the problem.
As such, it is sufficient to solve only the decision part of the problem 
as there already exists an efficient randomized algorithm for the reconstruction part.

\begin{lemma}[Theorem $3$ in \cite{ShpilkaVolkovich14}. Instantiation for the case $d=1$.] 
There is a polynomial-time randomized algorithm that
given oracle access to a read-one formula $\psi$
on $n$ variables, reconstructs $\psi$ with high probability.
If  $\size{\F} \leq  4n^2$, then the algorithm may make queries from an extension field of $\F$ of size larger than $4n^2$.
\end{lemma}

The randomized algorithm of \cite{ShpilkaVolkovich14} operates as follows.
Given oracle access to a polynomial $P$ run the reconstruction algorithm to get a candidate ROF $\psi$. 
If the reconstruction algorithm fails, we conclude that $P$ was not a ROP to begin with.
Otherwise, invoke the Schwartz-Zippel Lemma (Lemma \ref{lem:sz})
to check whether $\psi$ indeed computes $P$. This results in a two-sided error algorithm.
On one hand, given a ROP $P$ as an input the reconstruction algorithm may output a wrong ROF $\psi'$. On the other hand,
given a non ROP input $P$ the reconstruction algorithm may still output some ROF $\psi$ and there is a small chance that the Schwartz-Zippel algorithm 
will answer `yes' although there is no equality. We now give a simpler, one-sided error algorithm for the problem.

\begin{algorithm} 
\label{alg:ROT}
    \KwIn{$n,d \geq 1,\varepsilon > 0$, oracle access to $P \in \Fn$ of degree at most $d$.}
    \KwOut{``yes'' if $P$ is a ROP, ``no'' otherwise.}

    Pick $\ab \in \F^n$ at random \;
	 \ForEach{$I \subseteq [n]$ of size $\size{I} = 3$} {
        Set $P' \eqdef P \restrict{\xb_{[n] \setminus I} \; = \; \ab_{[n] \setminus I}}$ (interpolate are a trivariate polynomial of degree $d$ over $V$) \;
   	\If{$P'$ is not a multilinear polynomial}
    {Output ``no'' \;} 
    	\label{L5}
     \Else{Check if $P'$ is a ROP using Lemma \ref{lem:3ROP} \; }
     \label{L8}
	}
    Output ``yes'' iff $P$ passes all the tests \;

\caption{Read-Once Testing}
\end{algorithm}

\begin{lemma}
\label{lem:ROT}
Let $n,d \geq 1,\varepsilon > 0$ and $\F$ be a field of size $\size{\F} \geq \max \set{1.5n^4,d} / \varepsilon$. 
Given oracle access to a polynomial $P \in \Fn$ of a degree at most $d$
Algorithm \ref{alg:ROT} runs in time $\poly(n,d,1/ log(\varepsilon))$ and outputs ``yes'' if $P$ is a read-once polynomial.
Otherwise, the algorithm outputs ``no'' with probability at least $1-\varepsilon$.
\end{lemma}

\begin{proof}
The claim regarding the running time is immediate from the description of the algorithm.
Observe that the test in Line \ref{L8} is actually an identity test for a degree $2$ univariate polynomial.
For the correctness, first note that if $P$ is a ROP then it passes all the test for every $\ab \in \F^n$.
Now suppose that $P$ is not a ROP. As previously, we divide our analysis into two cases. \\ 
\textbf{Case $1$:} $P$ is not a multilinear polynomial.
Then there exists a variable $x_i$ and $e \geq 2$ such that $x^e_i$ appears in some monomial of $P$.
We can write $P = Q x^e_i + R$ where the degree of $x_i$ in $R$ is strictly less than $e$ (if any).
Now, pick $I \subseteq [n]$ of size $\size{I} = 3$ such that $i \in I$.
Note that if $Q \restrict{\bar{x}_{[n] \setminus I} = \bar{a}_{[n] \setminus I}} \nequiv 0$
then the corresponding $P'$ will fail the multilinearity test in Line \ref{L5}
By the Schwartz-Zippel Lemma (Lemma \ref{lem:sz}) 
$\Pr_{\; \bar{a} \in \F^n \;}[Q(\bar{a}) = 0] \leq \frac{d}{\size{\F}} < \varepsilon $ 
that $P$ will pass all tests with probability at most $\varepsilon$. \\ \\
\textbf{Case $2$:} $P$ is a multilinear polynomial.
As such, Lemma \ref{lem:3ROP} decides correctly whether or not $P'$ is a ROP.
Consequently, by Theorem \ref{THM:Main2}, $P$ could pass all tests for less than $\varepsilon$ fraction of $\bar{a} \in \F^n$,
which completes the proof.
\end{proof}

Theorem \ref{THM:Main4} follows as a corollary of the lemma by setting $\varepsilon = 1/4$.

\section{Lower Bounds \& Discussion}
\label{sec:lower}

As was mentioned earlier, our structural results require that the underlying field is of polynomial size.
This is common to many structural results for polynomials. In this section we try to complete the picture
by showing some lower bounds and impossibility results.
First, we exhibit lower bounds on the field size in Theorems \ref{THM:Main1}, \ref{THM:Main2} 
and Algorithm \ref{alg:ROT}.
Next, we show that similar structural statements are false over the Boolean domain, 
even for the monotone formulae.
We finish this section with some open questions.

\subsection{Lower Bounds on the Field Size}

Let $\F$ be a field. Consider the following family of multilinear polynomials.

\begin{definition} 
$\set{Q_n}_{n \in \N}: \F^n \to \F$, $Q_n(\xb) \eqdef \prod _{i=1}^n (x_i - 1) + \prod _{i=1}^n x_i$.
\end{definition}

First, observe that for $n \geq 3$ $Q_n$ is not a ROP for any field. We leave the proof as an exercise for the reader.
On the other hand, when $\F = \F_2$, the field of two elements, fixing even a single variable to any field element results in a ROP.
We get the following lemma.

\begin{lemma}
Let $n \geq 4$ and $\F = \F_2$.
Then there exists a multilinear read-many polynomial $Q_n \in \Fn$ such that $Q_n$ is \lROP for each $\bar{a} \in \F_2^n$.
\end{lemma}
This implies that field size in Theorem \ref{THM:Main1} should be at least $3$. 
We now move to the proof of a lower bound on the field size in Theorem \ref{THM:Main2} and Algorithm \ref{alg:ROT}.
For this purpose we need the following definition:

\begin{definition}
Let $S \subseteq \F$ and $\ab \in \F^n$. 
We define the size of $\ab$ w.r.t to $S$ as: $\size{\ab}_S \eqdef \size{\condset{i}{a_i \in S}}$. 
\end{definition}
We can now extend the previous result to other fields:

\begin{corollary}
Let $n \geq 4$ and $\F$ be a field.
Then there exists a multilinear read-many polynomial $Q_n \in \Fn$ such that $Q_n$ is \lROP for each $\bar{a} \in \F^n$
with  $\size{\ab}_{\bools{}} \geq 4$.
\end{corollary}

\begin{corollary}
Let $n \geq 4$ and $\F$ be a field of size $\size{\F} \leq n/4$.
Then there exists a multilinear read-many polynomial $Q_n \in \Fn$ such that $Q_n$ is \lROP for at least 
$1 - \exp(-n / \size{\F}^2)$ fraction of $\bar{a} \in \F^n$.
\end{corollary}

\begin{proof}
Let $ \frac{2}{\size{\F}} > \delta > 0 $.
By the Chernoff bound: $\Pr_{\ab \in \F^n} \left[ \frac{1}{n} \size{\ab}_{\bools{}} < \frac{2}{\size{\F}} - \delta \right] \leq \exp(- \delta^2 n)$.
Thus,   $\Pr_{\ab \in \F^n} \left[ \size{\ab}_{\bools{}} < 4 \right] \leq $
$\Pr_{\ab \in \F^n} \left[ \frac{1}{n} \size{\ab}_{\bools{}} < \frac{1}{\size{\F}} \right] \leq \exp(-n / \size{\F}^2)$.
\end{proof}

We can now give the lower bound.

\begin{corollary}
Let $n \geq 4,\varepsilon > 0$ and $\F$ be a field
such that every multilinear polynomial $P \in \Fn$ over $\F$ 
is a read-once polynomial if and only if $P$ is \lROP for at least $\varepsilon$ fraction of $\bar{a} \in \F^n$.
Then $\size{\F} = \Omega \left( \min \left(n, \sqrt{\frac{n}{\varepsilon}} \right) \right)$.
\end{corollary}

\begin{proof}
Suppose that $\size{\F} \leq n/4$.
Set $P=Q_n$. By the previous corollary: $\varepsilon > 1 - \exp(-n / \size{\F}^2)$, implying that
$\size{\F}^2 \geq \Omega(\frac{n}{-\ln(1-\varepsilon)}) = \Omega(\frac{n}{\varepsilon})$. 
Consequently, $\size{\F} = \Omega(\sqrt{\frac{n}{\varepsilon}})$.
\end{proof}

Moving to algorithmics, specifically considering Algorithm \ref{alg:ROT}, a standard way to reduce
the failure probability is by repeating the algorithm several times. 
Another corollary of the above analysis is that if the underlying field
is of size $\size{\F} = \BigO(n^{1/2 - \delta})$ for some $\delta > 0$ then the success probability of Algorithm \ref{alg:ROT} 
is exponentially small $\exp(-n^{2\delta})$. As a result, to reduce the failure probability below $\varepsilon$ one would need to
repeat the algorithm at least $\exp(n^{2\delta}) \cdot \ln(1/\varepsilon)$ times.

\subsection{Impossibility Results for Boolean Functions}

In \cite{Gurvich77,KLNSW93}, a characterization of functions computed by Boolean read-once formulae was given.
Those functions were referred to as ``read-once functions''.
The characterization was given in terms of minterms and maxterms of the Boolean functions in question.
The first step in this characterization was considering only monotone functions
\footnote{A Boolean function $f(\xb)$ is \emph{monotone} if for every $\xb \geq \yb \in \bools{n}$ it holds that $f(\xb) \geq f(\yb)$.}.
We show that statements similar to the ones proved in this paper
(i.e. ``global structure'' iff ``local structure'') are false over the Boolean domain, 
even if we restrict ourselves to the monotone functions. 
To this end, we define two families of Boolean functions.

\begin{definition}
$\space$ \\ 
$\set{f_n}_{n \in \N}: \bools{n} \to \bools{}$, 
$f_n(\xb) \eqdef x_1 \wedge x_2 \wedge \cdots \wedge x_n \bigvee \xb_1 \wedge \xb_2 \wedge \cdots \wedge \xb_n$. \\
$\set{g_n}_{n \in \N}: \bools{n+1} \to \bools{}$, 
$g_n(\xb,y) \eqdef y \wedge (x_1 \vee x_2 \vee \cdots \vee x_n) \bigvee x_1 \wedge x_2 \wedge \cdots \wedge x_n$.
\end{definition}

Observe that $\set{f_n}$ resembles $\set{Q_n}$ from the previous section
and in fact can be thought of a Boolean version of $Q_n$.
We get the following lemma:

\begin{lemma}
Let $n \geq 3$.
Then there exists a Boolean read-many function $f_n(x_1, \ldots, x_n)$ such that 
fixing any variable to either $0$ or $1$ results in a read-once function.
\end{lemma}
We now show a similar statement for monotone functions.

\begin{lemma}
Let $n \geq 2$.
Then there exists a monotone, read-many function $g_n(x_1, \ldots, x_n,y)$ such that 
fixing any variable to either $0$ or $1$ results in a monotone read-once function.
\end{lemma}

\begin{proof}
By a simple case analysis.  
\end{proof}
The above preclude any ``global structure'' iff ``local structure'' result in the Boolean domain
for any locality (not just three) even in the monotone setting.

\subsection{Open Questions}

We conclude with some open questions.
First of all, the previous sections exhibit some lower bounds of required field size.
It would be nice to get the right bound and see what is the behavior of such functions just below that bound.  

The other natural question is whether it is possible to get a characterization for functions computed by 
other interesting classes of functions, both Boolean and arithmetic?
Such as: read-twice formulae (or read-$k$ for $k\geq 2$) or even sum of two read-once formulae, bounded-depth formulae, etc. 
The same can be ask w.r.t property testers.

In \cite{ShpilkaVolkovich09}, it was shown that a sum of $k$ read-once polynomial $P_1 + \ldots + P_k$ is 
uniquely defined by its $\BigO(k)$-variate restrictions to a typical assignment.
The result was recently generalized in \cite{AvMV14} showing that a polynomial computed by multilinear read-$k$ is uniquely defined 
by its $k^{\BigO(k)}$-variate restrictions to a typical assignment.
In \cite{ShpilkaVolkovich14}, it was shown how to efficiently reconstruct a (single) read-once formula   
given the set of its three-variate restrictions to a typical assignment. 
However, for $k \geq 2$ the question of efficient reconstruction of multilinear read-$k$ formulae
remains open, even for special case of when the formula is a sum of read-once formula.

Giving a characterization can be viewed as an intermediate task.
So, we finish with a conjecture which can be seen as an extension of Theorem \ref{THM:Main2}:
``There exists a function $\ell oc(k) : \N \to \N$ such that a polynomial $P$ is computable by a multilinear read-$k$ formula
iff the same holds true for each of its restriction of size $\ell oc(k)$ to a typical assignment''.



\newcommand{\etalchar}[1]{$^{#1}$}

\end{document}